\documentclass[journal,onecolumn,10pt]{IEEEtran}
\usepackage{schemabloc}
\usepackage{amsthm}
\usepackage[cmex10]{amsmath}
\usepackage{amsfonts}
\usepackage{footnote}
\usepackage{cite}
\usepackage{graphicx, subfigure}
\usepackage{setspace}
\usepackage[latin1]{inputenc}
\usepackage{tikz}
\usetikzlibrary{chains,scopes}
\usetikzlibrary{shapes,arrows}
\usepackage{lineno}
\usepackage{cite}
\usepackage{stackrel}

\newtheorem{lemma}{Lemma}

\begin{document}
\allowdisplaybreaks
\title{On the Information in Extreme Measurements for Parameter Estimation}

\author{\IEEEauthorblockN{Jonatan Ostrometzky\IEEEauthorrefmark{1} and Hagit Messer\IEEEauthorrefmark{1}~\IEEEmembership{Fellow,~IEEE}\\} 
\IEEEauthorblockA{\IEEEauthorrefmark{1}Department of Electrical Engineering, Columbia University in the City of New York \\ \IEEEauthorrefmark{2}School of Electrical Engineering,
Tel~Aviv University}
}


\maketitle


\begin{abstract}
This paper deals with parameter estimation from extreme measurements. While being a special case of parameter estimation from partial data, in scenarios where only one sample from a given set of K measurements can be extracted, choosing only the minimum or the maximum (i.e., extreme) value from that set is of special interest because of the ultra-low energy,
storage, and processing power required to extract extreme values from a given data set. We present a new methodology to analyze the performance of parameter estimation from extreme measurements. In particular, we present a general close-form approximation for the Cramer-Rao Lower Bound on the parameter estimation error, based on extreme values. We demonstrate our methodology on the case where the original measurements are exponential distributed, which is related to many practical applications. The analysis shows that the maximum values carry most of the information about the parameter of interest and that the additional information in the minimum is negligible. Moreover, it shows that for small sets of iid measurements (e.g. K=15) the use of the maximum can provide data compression with factor 15 while keeping about 50\% of the information stored in the complete set. 

\end{abstract}


\onehalfspacing

\section{Introduction}
\IEEEPARstart{N}{on}-Bayesian parameter estimation is usually based on the assumption that an observation vector (defined by $\underline{x}$) is available, and that the Probability Density Function (PDF) $f_{\underline{X}}(\underline{x};\underline{\theta})$ is known and is dependent on the parameter vector $\underline{\theta}$, to be estimated \cite{Key}. When dealing with the case where the observations are identically and independent distributed (iid), knowing the PDF $f_{X_i}(x_i;\underline{\theta})$ of each observation $x_i$ is sufficient. However, there are certain situations in which the full observation vector is unavailable or unobservable. Specifically, in this paper, we focus on situations where from a set of observations, only the extreme values, i.e., the minimum and/or the maximum observations are reported.


We deal with a sequence of $N\cdot{K}$ iid observations. Define a sub-sequence of $K$ iid observations which constitute the $i^{th}$ group (i.e., interval), by the vector $\underline{x}_i$, and the minimal and the maximal observed values in the $i^{th}$ interval by $y_{min_i}\equiv{min(\underline{x}_i)}$ and $y_{max_i}\equiv{max(\underline{x}_i)}$, respectively. Thus, the minimum-observation vector of $N$ non-overlapping intervals is defined by $\underline{y}_{min}=[y_{{min}_1},y_{{min}_2},\cdots,y_{{min}_N}]^T$ and the maximum-observation vector of the same $N$ intervals is defined by $\underline{y}_{max}=[y_{{mx}_1},y_{{max}_2},\cdots,y_{{max}_N}]^T$. The complete (unobservable) data is the $N\cdot{K}$-dimension observation vector which can be formulated as: $\underline{x}=[\underline{x}^T_1,\underline{x}^T_2,\cdots,\underline{x}^T_N]^T$.

Such scenarios, where only $\underline{y}_{min}$ and/or $\underline{y}_{max}$ are given rather than $\underline{x}$, can be found in numerous fields, such as earthquake recurrence estimation \cite{wang2011}, wind research \cite{walshaw1994getting}, heat accumulation \cite{MaxTemp1}, and precipitation monitoring \cite{ExtremeRainEx1}, to name a few. Furthermore, there are practical reasons which may force compressing the entire dataset to the $\underline{y}_{min}$ and/or $\underline{y}_{max}$ observations, such as energy constraints in sensor networks \cite{EnvSN}. E.g., most Network Management Systems (NMS) which monitor the backhaul of the cellular networks, although sample the network microwave links channels Received Signal Levels (RSL) at a high frequency (of up to 10Hz), usually report only the minimum and the maximum observed values per 15-minute intervals \cite{RemkoCountry2,Eric10S,yonidiss}. Furthermore, reporting only the extreme values from a given set of observations is especially attractive in the emerging field of the Internet-of-Things (IoT), since, apart from transmission costs, it has been shown that the extraction of the minimum and the maximum values from a sequence can be done extremely efficient with respect to processing power, delay, energy, and memory requirements \cite{lemire2006streaming}. Indeed, there are scenarios where it could be more beneficial to treat other types of information rather than the minimum and/or the maximum (e.g., when dealing with the Normal distribution, calculating the \emph{mean} and the \emph{standard deviation} from the original observations may be preferable, since they are sufficient statistics). However, one still needs to consider the fact that under extreme energy and cost constraints, the additional hardware needed for these kinds of calculations may not be available. The minimum and the maximum values, on the other hand, as they are the extremes of the original observation vector, can be identified using very basic circuits, without the need for extra processing power \cite{sedraandsmith}. 

Thus, the motivation to understand the estimation accuracy that can be achieved by using only extreme measurements (in comparison with the full original sample set) is strong, especially as giving access only to extreme values may result in much more efficient designs of future applications, if the achievable estimation accuracy is sufficient.

In this paper we show that it is possible to evaluate the accuracy of the estimate of a parameter vector $\underline{\theta}$ given only the minimum, the maximum, or the minimum \emph{and} the maximum observation vectors ($\underline{y}_{min}$, $\underline{y}_{max}$, and \{$\underline{y}_{min}$,$\underline{y}_{max}$\}, respectively). We describe the corresponding Fisher Information Matrices (FIM) of the (asymptotically) optimal estimators in the Minimum Mean Squared Error sense (MMSE), using a novel approximation which simplify the resulted expressions so that they become analytically solvable, and, by means of the approximated Cramer-Rao Lower Bound (CRLB), we present an analysis of the achievable estimation performance. 

\subsection{Related Studies}
Many previous studies have dealt with the problem of parameter estimation involving $\underline{y}_{min}$ or $\underline{y}_{max}$. The subject of most of these studies, however, was to characterise the extreme values themselves. I.e., these studies aimed to estimate the properties of extreme events (such as catastrophic earthquakes probabilities \cite{wang2011}, variations in climate \cite{ClimateExtreme92}, extreme floods \cite{Flood2007,Katz3}, long precipitation events \cite{ExtremeRainEx1}, daily rainfall \cite{coles2003anticipating}, etc.), usually by taking advantage of the Extreme Value Theory (EVT). The EVT had started to attract interest in the last century, and was formalized by \emph{Gumbel} in 1958 \cite{Gumbel}. In its base, the EVT states that under some regularity conditions, the PDF of maximum (or minimum, sometimes under certain transformations) values converges asymptotically to the Generalized Extreme Value (GEV) PDF, $f_{\underline{Y}}(\underline{y};\underline{\psi})$ \cite{FisherT,walshaw2013}. Thus, it is possible to estimate the parameter vector $\underline{\psi}$ \cite{GEVEstimator} (or other properties such as moments \cite{hosking1985,Dekkers}), directly from $\underline{y}_{min}$,$\underline{y}_{max}$.

In recent years, the estimation of the parameter vector of the original PDF, $\underline\theta$, from extreme values has been partially covered by the \emph{Record Theory}, which deals with ordered data \cite{records}. And indeed, numerous studies discussed the estimation of $\underline\theta$ from maximum or minimum values \cite{recordsparam1986}, \cite{records}, as well as the FIM properties of those estimates \cite{ahmadi2001fisher,hofmann2003,recordsreview2009}. However, as described in those studies, the presented FIM expressions, apart for some specific examples such as for the exponential distribution case \cite{hofmann2003}, are implicit and may not hold a simple solution \cite{recordsreview2009}. Recently, we have shown that an estimate $\hat{\underline{\theta}}$ of the original parameter vector $\underline{\theta}$ can be evaluated directly from the GEV estimated parameter vector $\hat{\underline{\psi}}$ \cite{SAM2014}. This method, however, relies on the asymptotic convergence of the PDF of the maximum values to the GEV, and thus, cannot guarantee optimal performance in non-asymptotic conditions. Several other studies presented approaches for the estimation of $\underline{\theta}$ from incomplete observations-set, such as the Expectation-Maximization (EM) algorithm \cite{EM}, for which $\underline{x}$ is the complete information and $\underline{y}_{min}$,$\underline{y}_{max}$ are the incomplete information. However, the estimator proposed in these studies basically extant approximation of the full observation vector $\hat{\underline{x}}$ from the set of the incomplete data (i.e., the extremes $\underline{y}_{min}$,$\underline{y}_{max}$), and estimate $\underline{\theta}$ from $\hat{\underline{x}}$. The EM, as other methods which first recover $\hat{\underline{x}}$ (or its properties) and then estimate $\underline{\theta}$ assuming $\hat{\underline{x}}$ is the observation vector \cite{MaxTemp1,LibermanTh}, are usually case specific, often require pre-calibration stages \cite{RemkoCountry,Yoni1}, and are sub-optimal due to the approximation of $\underline{x}$.

Different from these (among many) past studies, in this paper we are interested in evaluating the achievable performance of the estimation of the parameter vector $\underline{\theta}$ of the \emph{original} PDF, directly from the available extreme measurements. We present a tool which approximates the cumbersome expressions of the relevant FIM into a simple, practical, and solvable form, without a significant loss of accuracy. Furthermore, we extend previous results and discuss the case where both the minimum \emph{and} the maximum values are being used in the same estimation process.

\subsection{Summary of the Results}
We consider the (asymptotically) optimal estimates of the parameter vector ,$\underline\theta$ based on the set of the extremes $\underline{y}_{min}$,$\underline{y}_{max}$, and we present a novel and simple performance analysis tools. We consider three cases:
\begin{itemize}
\item[1.]{Only the minimum-observation vector, $\underline{y}_{min}$, is available.}
\item[2.]{Only the maximum-observation vector, $\underline{y}_{max}$, is available.}
\item[3.]{Both the minimum-observation vector, $\underline{y}_{min}$, \emph{and} the maximum-observation vector, $\underline{y}_{max}$, are available.}
\end{itemize}

For each of the three cases we derive the appropriate FIM, and present simplified approximated expressions, to be used in the comparative performance analysis. The optimal estimator, $\hat{\underline{\theta}}_{opt}$, makes use of the entire range of $K\cdot{N}$ samples, once the observation vector $\underline{x}$ is available. Using only a part of the available samples results in a sub-optimal estimation, and thus, decreases the estimation accuracy \cite{Key}. While it is obvious that any estimation based on $\underline{y}_{min}$ and/or $\underline{y}_{max}$ is suboptimal, our aim is to quantify the relative performance loss of the estimators based on $\underline{y}_{min}$ and/or $\underline{y}_{max}$, relative to the performance of $\hat{\underline{\theta}}_{opt}$.


We show that whereas using both the minimum and the maximum measurements in the estimation is preferred (regarding the estimates accuracy), for certain distributions, most of the information about $\underline{\theta}$ is stored in one of the extremes (i.e., only the minimum or only the maximum values contain most of the information needed for the estimation). And, by using the tools we present, it is possible to pre-determine which of the extremes is preferable, and how much information it contains. We validate our findings by a series of simulations, compare the results with previously published studies for the exponential case, where analytical solutions for the FIM expressions exist \cite{records,hofmann2003}, and show that our proposed approximations are valid. 

Thus, the main contributions of this paper are:
\begin{enumerate}
\item{Presentation of a novel approximation for the FIM of estimates based on extreme measurements (minimum, maximum, \emph{and} minimum and maximum combined) by using the Characteristic Values of the Extremes, in order to express an approximate CRLB analytically;}
\item{Analysis of extreme measurements in terms of the information they hold (with respect to parameter estimation) compared with the original set of samples.} 
\end{enumerate}

\subsection{Organization} 
The rest of the paper is organized as follows: In \emph{Section II} the methodology, tools, and our analysis are presented. \emph{Section III} presents an explicit development of our tools, followed by a simulation which is used in order to validate the results. A comparison with previously published results is also included and discussed in this section. Lastly, in \emph{Section V} we conclude this paper.

\section{Methodology and Tools}
The asymptotically \emph{optimal} (in the MMSE sense) estimator $\hat{\underline{\theta}}_{opt}$ is given by:
\begin{IEEEeqnarray}{cCl}
\hat{\underline{\theta}}_{opt}=\arg\max\limits_{\underline{\theta}}\left\{\sum_{i=1}^{N\cdot{}K}f_{{X}}({x_i};\underline{\theta})\right\} \IEEEyesnumber \label{MLEOpt} 
\IEEEyesnumber
\end{IEEEeqnarray}
where ${x}_i$ are the original iid observations (i.e., measurements), and $\underline{\theta}$ is the parameter vector, to be estimated. And, under mild regularity conditions, $\hat{\underline{\theta}}_{opt}$ achieves the corresponding Cramer-Rao Lower Bound (CRLB) \cite{Key}, defined by $\text{CRLB}_{opt}$:
\begin{IEEEeqnarray}{cCl} \label{CRLBOpt}
\text{CRLB}_{opt}=J^{opt}(\underline{\theta})^{-1} \IEEEyesnumber \IEEEyessubnumber \\
{J_{m,n}^{opt}(\underline{\theta})}=-\sum_{i=1}^{N\cdot{}K}{E\left\{\frac{\partial^2}{\partial{\theta}_m\partial{\theta}_n}\log\left[f_{{X}}({x_i};\underline{\theta})\right]\right\}} \IEEEyessubnumber
\end{IEEEeqnarray}
where $J^{opt}(\underline{\theta})$ is the FIM consisting of the entries $J_{m,n}^{opt}(\underline{\theta})$. 

In addition to the optimal estimator, we define a sub-optimal estimator, $\hat{\underline{\theta}}_L$, based on partial samples $L\in{\mathbb{N}}$, s.t. $1\leq{L}\leq{K}$ observations from each of the $N$ non overlapping sub-intervals:
\begin{IEEEeqnarray}{cCl}
\hat{\underline{\theta}}_{L}=\arg\max\limits_{\underline{\theta}}\left\{\sum_{i=1}^{N}\sum_{j=1}^{{L}}f_{{X}}({x_{i,j}};\underline{\theta})\right\} \IEEEyesnumber \label{MLEl} \IEEEyesnumber
\end{IEEEeqnarray}
which achieves
\begin{IEEEeqnarray}{cCl}
\text{CRLB}_{L}=J^{L}(\underline{\theta})^{-1} \IEEEyesnumber \label{CRLBl}  \IEEEyessubnumber \label{CRLBla} \\ 
{J_{m,n}^{L}(\underline{\theta})}=-\sum_{i=1}^{N}\sum_{j=1}^{{L}}{E\left\{\frac{\partial^2\log\left[f_{{X}}({x}_{i,j};\underline{\theta})\right]}{\partial{\theta}_m\partial{\theta}_n}\right\}}  \IEEEyessubnumber \label{CRLBlb}
\end{IEEEeqnarray}
where ${x}_{i,j}$ represents the partial measurements, and $J^{L}(\underline{\theta})$ is the FIM (of the entries $J_{m,n}^{L}(\underline{\theta})$). This sub-optimal estimator $\hat{\underline{\theta}}_L$ and its corresponding CRLB ($\text{CRLB}_L$) will be used as a performance benchmark tool in the sequel.


\subsection{$\underline{y}_{min}$ and/or $\underline{y}_{max}$ based Estimation}
The PDF of an extreme value taken from the $K$ measurements which constitute the $i^{th}$ interval, $y_{min_i}$ and $y_{max_i}$, and the joint PDF of $y_{min_i}$ and $y_{max_i}$ are given by \cite{Gumbel}:
\begin{IEEEeqnarray}{cCl} \label{fgen}
f_{Y_{min}}(y_{min};\underline{\theta})= \IEEEyesnumber \IEEEnonumber 
K\left[1-F_{X}(y_{min};\underline{\theta})\right]^{K-1}f_{X}(y_{min};\underline{\theta}) \IEEEyessubnumber \label{fmin} \\
f_{Y_{max}}(y_{max};\underline{\theta})=K\left[F_{X}(y_{max};\underline{\theta})\right]^{K-1}f_{X}(y_{max};\underline{\theta}) \IEEEyessubnumber \label{fmax} \\
f_{Y_{min},Y_{max}}(y_{min},y_{max};\underline{\theta})=
K(K-1)\left[F_{X}(y_{max};\underline{\theta})-F_{X}(y_{min};\underline{\theta})\right]^{K-2}\cdot{} \IEEEyessubnumber 
f_{X}(y_{min};\underline{\theta})f_{X}(y_{max};\underline{\theta}) 
\end{IEEEeqnarray}
where $F(\cdot)$ represents a Cumulative Density Function (CDF). 

Due to the iid properties of $\underline{x}$, which guarantee that the extreme values are also iid\footnote{Meaning, that the set of the maximum values are iid, and the set of the minimum values are iid. However, for each sub-interval, the minimum value and the maximum value are dependant for any $K<\infty$ \cite{coles1999dependence}.}, the PDF of $\underline{y}_{min}$, $\underline{y}_{max}$ and the joint PDF of $\underline{y}_{min}$ and $\underline{y}_{max}$ can be easily expressed:
\begin{IEEEeqnarray}{cCl} \label{fgeni}
f_{\underline{Y}_{min}}(\underline{y}_{min};\underline{\theta})= \IEEEyesnumber \IEEEnonumber 
{K}^N\prod_{i=1}^{N}\left[1-F_{X}(y_{min_i};\underline{\theta})\right]^{K-1}f_{X}(y_{min_i};\underline{\theta}) \IEEEyessubnumber \\
\label{fmini}
f_{\underline{Y}_{max}}(\underline{y}_{max};\underline{\theta})= \IEEEnonumber 
{K}^N\prod_{i=1}^{N}\left[F_{X}(y_{max_i};\underline{\theta})\right]^{K-1}f_{X}(y_{max_i};\underline{\theta}) \IEEEyessubnumber \\
\label{fmaxi}
f_{\underline{Y}_{min},\underline{Y}_{max}}(\underline{y}_{min},\underline{y}_{max};\underline{\theta})= \left[{K(K-1)}\right]^N\cdot{} 
\prod_{i=1}^{N}\left[F_{X}(y_{max_i};\underline{\theta})-F_{X}(y_{min_i};\underline{\theta})\right]^{K-2}\cdot{} \IEEEyessubnumber 
f_{X}(y_{min_i};\underline{\theta})f_{X}(y_{max_i};\underline{\theta})  \label{fminmaxi}
\end{IEEEeqnarray}

From which, the desired estimators, can be directly formulated. For the sequence, we define
 $\hat{\underline{\theta}}_{min}$, $\hat{\underline{\theta}}_{max}$, and $\hat{\underline{\theta}}_{mix}$ as the estimates based on $\underline{y}_{min}$, $\underline{y}_{max}$, and \{$\underline{y}_{min}$,$\underline{y}_{max}$\}, respectively. 
 
The corresponding FIM of these estimates is given by:
\begin{IEEEeqnarray}{cCl} 
{J^{min}_{m,n}(\underline{\theta})}=-{E\left\{\frac{\partial^2}{\partial{\theta}_m\partial{\theta}_n}\log\left[f_{\underline{Y}_{min}}(\underline{y}_{min};\underline{\theta})\right]\right\}}
=-\sum_{i=1}^{N}{E\left\{\frac{\partial^2}{\partial{\theta}_m\partial{\theta}_n}\log\left[f_{{Y}_{min_i}}({y_{min_i}};\underline{\theta})\right]\right\}} \IEEEyesnumber \label{FIMmixt}\IEEEyessubnumber \label{CRLBmi} \\
{J^{max}_{m,n}(\underline{\theta})}=-{E\left\{\frac{\partial^2}{\partial{\theta}_m\partial{\theta}_n}\log\left[f_{\underline{Y}_{max}}(\underline{y}_{max};\underline{\theta})\right]\right\}}
=-\sum_{i=1}^{N}{E\left\{\frac{\partial^2}{\partial{\theta}_m\partial{\theta}_n}\log\left[f_{{Y}_{max_i}}({y}_{max_i};\underline{\theta})\right]\right\}} \IEEEyessubnumber \label{CRLBmx} \\
{J^{mix}_{m,n}(\underline{\theta})}=-{E\left\{\frac{\partial^2\log\left[f_{\underline{Y}_{{min}},\underline{Y}_{{max}}}(\underline{y}_{min},\underline{y}_{max};\underline{\theta})\right]}{\partial{\theta}_m\partial{\theta}_n}\right\}}
=-\sum_{i=1}^{N}{E\left\{\frac{\partial^2\log\left[f_{{Y}_{min_i},Y_{max_i}}({y}_{{min_i}},y_{max_i};\underline{\theta})\right]}{\partial{\theta}_m\partial{\theta}_n}\right\}} \IEEEyessubnumber \label{CRLBmix}
\end{IEEEeqnarray}

And, as in \eqref{CRLBOpt}, under mild regularity conditions, the asymptotic covariance matrices of the estimates $\hat{\underline{\theta}}_{min}$, $\hat{\underline{\theta}}_{max}$, and $\hat{\underline{\theta}}_{mix}$ achieve the corresponding CRLBs:
\begin{IEEEeqnarray}{cCl} \label{CRLBmixt}
\text{CRLB}_{min}=J^{min}(\underline{\theta})^{-1} \IEEEyesnumber \IEEEyessubnumber \\
\text{CRLB}_{max}=J^{max}(\underline{\theta})^{-1} \IEEEyessubnumber \\
\text{CRLB}_{mix}=J^{mix}(\underline{\theta})^{-1} \IEEEyessubnumber
\end{IEEEeqnarray}
where $\text{CRLB}_{min}$, $\text{CRLB}_{max}$, and $\text{CRLB}_{mix}$, are the corresponding asymptotic covariances of $\hat{\underline{\theta}}_{min}$, $\hat{\underline{\theta}}_{max}$, and $\hat{\underline{\theta}}_{mix}$, respectively.

In order to simplify the presentation of our approach, the parameter vector $\underline\theta$ is reduced to a single parameter (i.e., $\underline\theta\equiv\theta$). The same development stages and conclusions presented in the sequel can be generalized for the multi-parameter case.

\subsection{The Relationship Between $J^{min}(\theta)$ and $J^{max}(\theta)$}
While it is obvious that the following relationship stands for any distribution:
\begin{IEEEeqnarray}{cCl} 
\left\{J^{min}(\theta),{J}^{max}(\theta)\right\} \label{MajorG} 
\leq{J}^{mix}(\theta)\leq{J}^{opt}(\theta) \IEEEyesnumber
\end{IEEEeqnarray}
the specific order of $J^{min}(\theta)$ and $J^{max}(\theta)$ is distribution dependent.

Next, we introduce the following:
\begin{lemma}
\label{Lemma1}
Let $\underline{z}$ be a ${K}$-dimensional observation vector, the entries of which are iid with PDF $f_Z(z;\underline\theta)$, of \textbf{finite} sample space of $\Omega=[C_1,C_2]$, meaning that the entries of $\underline{z}$ are bounded so that $\forall{i}:C_1\leq{z_i}\leq{C_2}$. Define two random variables, $Y_{min}$ and $Y_{max}$, so that ${Y}_{min}$=$min(\underline{z})$ and ${Y}_{max}$=$max(\underline{z})$. For any given function $g(\cdot)$ that satisfies a regularity condition so that $\forall \alpha: \left|\frac{\partial{^i}{g(\alpha)}}{\partial{\alpha}{^i}}\right|<\infty;i\in{\mathbb{N}}$, the following hold:
\begin{IEEEeqnarray}{cCl}
g(Y_{min})\xrightarrow[K\rightarrow\infty]{{w.p.1}}g(E[Y_{min}]) \IEEEyesnumber \IEEEyessubnumber  \label{eqtaylormindeducted} \\
g(Y_{max})\xrightarrow[K\rightarrow\infty]{{w.p.1}}g(E[Y_{max}]) \IEEEyessubnumber \label{eqtaylormaxdeducted}
\end{IEEEeqnarray}
\end{lemma}
\begin{proof}
Eqs. \eqref{fmin} and \eqref{fmax}, in combination with the finite sample space $\Omega$, yield
\begin{IEEEeqnarray}{cCl}
Y_{min}\xrightarrow[K\rightarrow\infty]{{w.p.1}}C_1 \IEEEyesnumber \IEEEyessubnumber \label{eqyminc1} \\
Y_{max}\xrightarrow[K\rightarrow\infty]{{w.p.1}}C_2 \IEEEyessubnumber \label{eqymaxc2}
\end{IEEEeqnarray}
The Taylor expansions of $g(Y_{min})$ about $E[Y_{min}]$ and $g(Y_{max})$ about $E[Y_{max}]$ are
\begin{IEEEeqnarray}{cCl}
g(Y_{min})=g(E[Y_{min}])+ \IEEEyesnumber 
\sum_{i=1}^{\infty}\frac{{(Y_{min}-E[Y_{min}])}^i}{i!}\cdot{}\left.\frac{\partial^i{g}(Y_{min})}{\partial{Y^i_{min}}}\right|_{E[Y_{min}]} \IEEEyessubnumber \label{eqtaylormin} \\
g(Y_{max})=g(E[Y_{max}])+
\sum_{i=1}^{\infty}\frac{{(Y_{max}-E[Y_{max}])}^i}{i!}\cdot{}\left.\frac{\partial^i{g}(Y_{max})}{\partial{Y^i_{max}}}\right|_{E[Y_{max}]} \IEEEyessubnumber \label{eqtaylormax}
\end{IEEEeqnarray}
Applying the \emph{Expected Value} operator on Eqs. \eqref{eqyminc1} and \eqref{eqymaxc2} yields 
\begin{IEEEeqnarray}{cCl}
E[Y_{min}]\xrightarrow[K\rightarrow\infty]{{w.p.1}}C_1 \IEEEyesnumber \IEEEyessubnumber \label{eqexpmin} \\
E[Y_{max}]\xrightarrow[K\rightarrow\infty]{{w.p.1}}C_2 \IEEEyessubnumber \label{eqexpmax}
\end{IEEEeqnarray}
which dictates that the sums of Eqs. \eqref{eqtaylormin} and \eqref{eqtaylormax} converge to zero:
\begin{IEEEeqnarray}{cCl}
\sum_{i=1}^{\infty}\frac{{(Y_{min}-E[Y_{min}])}^i}{i!}\cdot{} \IEEEyesnumber  
\left.\frac{\partial^i{g}(Y_{min})}{\partial{Y^i_{min}}}\right|_{E[Y_{min}]}\xrightarrow[K\rightarrow\infty]{{w.p.1}}0 \IEEEyesnumber \IEEEyessubnumber \\
\sum_{i=1}^{\infty}\frac{{(Y_{max}-E[Y_{max}])}^i}{i!}\cdot{} 
\left.\frac{\partial^i{g}(Y_{max})}{\partial{Y^i_{max}}}\right|_{E[Y_{max}]}\xrightarrow[K\rightarrow\infty]{{w.p.1}}0 \IEEEyessubnumber
\end{IEEEeqnarray}
This completes the proof.
\end{proof}

\begin{lemma}
\label{Lemma2}
Let $\underline{z}$ be a ${K}$-dimensional observation vector, the entries of which are iid with PDF $f_Z(z;\underline\theta)$, of \textbf{infinite} sample space of $\Omega=(-\infty,\infty)$, meaning that the entries of $\underline{z}$ are unbounded so that $\forall{i}:-\infty{}\leq{z_i}\leq{\infty}$. Define two random variables, $Y_{min}$ and $Y_{max}$, so that ${Y}_{min}$=$min(\underline{z})$ and ${Y}_{max}$=$max(\underline{z})$ each converges in distribution to an asymptotic GEV PDF with shape parameter $\epsilon{<}{\frac{1}{2}}$. Then, the following hold:
\begin{IEEEeqnarray}{cCl}
\frac{Y_{min}}{E[Y_{min}]}\xrightarrow[K\rightarrow\infty]{{a.s.}}1 \IEEEyesnumber \IEEEyessubnumber \\ \label{eqlemma2min}
\frac{Y_{max}}{E[Y_{max}]}\xrightarrow[K\rightarrow\infty]{{a.s.}}1 \IEEEyessubnumber \label{eqlemma2max}
\end{IEEEeqnarray}
\end{lemma}
\begin{proof}
Define $err_{min}\equiv{Y_{min}}-E[Y_{min}]$ and $err_{max}\equiv{Y_{max}}-E[Y_{max}]$. Based on the properties of the GEV \cite{Gumbel}, two finite constants exist (defined by $D_1$ and $D_2$), so that
\begin{IEEEeqnarray}{cCl}
\lim_{K \to \infty} \left\{var(err_{min})\right\}=D_1 \IEEEyesnumber \IEEEyessubnumber \label{errvarmin} \\
\lim_{K \to \infty} \left\{var(err_{max})\right\}=D_2 \IEEEyessubnumber \label{errvarmax}
\end{IEEEeqnarray}
which, based on Chebyshev's inequality \cite{papoulis2002probability}, in combination with the fact that
\begin{IEEEeqnarray}{cCl}
E[Y_{min}]\xrightarrow[K\rightarrow\infty]{{}}{-\infty} \IEEEyesnumber \IEEEyessubnumber \label{errmin} \\
E[Y_{max}]\xrightarrow[K\rightarrow\infty]{{}}\infty \IEEEyessubnumber \label{errmax}
\end{IEEEeqnarray}
completes this proof.
\end{proof}

Under the conditions of \emph{Lemma} \ref{Lemma1}, the variances of $Y_{min}$ and $Y_{max}$ converge to zero (as $K$ increases). Thus,
\begin{IEEEeqnarray}{cCl}
E[g(Y_{min})]={g}(E[Y_{min}])+\mathcal{R}^{min} \IEEEyesnumber \IEEEyessubnumber  \\
E[g(Y_{max})]={g}(E[Y_{max}])+\mathcal{R}^{max} \IEEEyessubnumber
\end{IEEEeqnarray}
where $\mathcal{R}^{min}$ and $\mathcal{R}^{max}$ are the residues, which converge to zero as $K$ increases.

\emph{Lemma} \ref{Lemma2} is weaker than \emph{Lemma} \ref{Lemma1} in the sense that under the conditions of \emph{Lemma} \ref{Lemma2}, the variances of $Y_{min}$ and $Y_{max}$ converge to known constants (as $K$ increases), and not to zero. Thus, under the conditions of \emph{Lemma} \ref{Lemma2}, the residuals $\mathcal{R}^{min}$ and $\mathcal{R}^{max}$ would not converge to zero. On the other hand, these residues do not depend on $K$, whereas the expected values of the extremes do \cite{Gumbel}. Thus, given a function $g(\cdot)$ that satisfies a regularity condition so that $\forall \alpha: \left|\frac{\partial{^i}{g(\alpha)}}{\partial{\alpha}{^i}}\right|<\infty;i\in{\mathbb{N}}$
\begin{IEEEeqnarray}{cCl}
E[g(Y_{min})]=\underbrace{{g}(E[Y_{min}])}_{\mathcal{O}\left((ln(K)\right)}+\underbrace{\mathcal{R}^{min}}_{\mathcal{O}\left(1\right)} \IEEEyesnumber \IEEEyessubnumber  \\
E[g(Y_{max})]=\underbrace{{g}(E[Y_{max}])}_{\mathcal{O}\left((ln(K)\right)}+\underbrace{\mathcal{R}^{max}}_{\mathcal{O}\left(1\right)} \IEEEyessubnumber \label{resmaxappdx1}
\end{IEEEeqnarray}
and for sufficiently large values of $K$, the following conclusions can be expressed:

\begin{IEEEeqnarray}{cCl}
E[g(Y_{min})]={g}(E[Y_{min}])\cdot{}\left(1+\mathcal{O}\left(\frac{1}{\ln{(K)}}\right)\right) \IEEEyesnumber \IEEEyessubnumber \label{ConclusionGminAp} \\
E[g(Y_{max})]={g}(E[Y_{max}])\cdot{}\left(1+\mathcal{O}\left(\frac{1}{\ln{(K)}}\right)\right) \IEEEyessubnumber \label{ConclusionGmaxAp}
\end{IEEEeqnarray}
which results in our main approximation:
\begin{IEEEeqnarray}{cCl}
E[g(Y_{min})]\approx{g}(E[Y_{min}]) \IEEEyesnumber \IEEEyessubnumber \label{ConclusionGminAp} \\
E[g(Y_{max})]\approx{g}(E[Y_{max}]) \IEEEyessubnumber \label{ConclusionGmaxAp}
\end{IEEEeqnarray}
for both finite and infinite sample spaces. Furthermore, the same conclusions hold in cases where the sample space is infinite, but is one-side bounded, e.g., the \emph{Gamma} distribution, the sample space of which is $\Omega=[0,\infty)$.

In addition, these conclusions can be easily expanded for the two-variables case, resulting in, assuming $g(\cdot)$ satisfies the same regularity condition of \emph{Lemma} \ref{Lemma1} and/or \emph{Lemma} \ref{Lemma2} for every variable, and sufficiently large $K$,
\begin{equation} \label{ConclusionGmixApp}
E[g(Y_{min},Y_{max})]\approx{g}(E[Y_{min}],E[Y_{max}])
\end{equation}
as for $K\rightarrow\infty$, the minimum and the maximum values taken from the same group become independent \cite{coles1999dependence}.

\subsection{Characteristic Values of Extremes}
In order to specify Eq. \eqref{MajorG} for a given distribution, we suggest to use the Characteristic Largest (or Smallest) Value \cite{Gumbel}.

The Characteristic Largest Value (CLV) and the Characteristic Smallest Value (CSV) were first introduced by \emph{Gumbel} in 1958 \cite{Gumbel}, as averages of extremes that are analogous to quantiles: the CLV is defined as $\mu_K$, so that from $K$ iid observations, exactly one observation is expected to be equal to or larger than $\mu_K$. In a similar manner, the CSV is defined as $\mu_1$, so that from $K$ iid observations, exactly one observation is expected to be equal to or smaller than $\mu_1$. The CSV and the CLV definitions lead to the relations
\begin{IEEEeqnarray}{cCl}
F_X(\mu{_1};\underline{\theta},K)=\frac{1}{K} \IEEEyesnumber \label{CLVG} \IEEEyessubnumber \label{CSVch2} \\
F_X(\mu{_K};\underline{\theta},K)=1-\frac{1}{K} \IEEEyessubnumber \label{CLVch2}
\end{IEEEeqnarray}
This means that, given a CDF of the original observation, $F_X({x;{\underline\theta}})$, which is analytically representable with respect to $K$, $F_X({x;\underline\theta,K})$, the explicit expressions of the CSV, $\mu_1$, and the CLV, $\mu_K$, can be directly extracted.

The characteristic largest and smallest values are closely related to the expected values of the extremes. In cases where the sample space is finite, the CLV and/or the CSV
converge to the boundaries of the sample space as $K$ increases, and thus, converge to the expected value of the extreme. Hence, for any given function $g(\cdot)$ that satisfies the regularity condition $\forall \alpha: \left|\frac{\partial{^i}{g(\alpha)}}{\partial{\alpha}{^i}}\right|<\infty;i\in{\mathbb{N}}$, the following hold:
\begin{IEEEeqnarray}{cCl} \label{ConclusionULem1}
E[g(Y_{min})]\xrightarrow[K\rightarrow\infty]{}{g}(\mu_1) \IEEEyesnumber \IEEEyessubnumber \label{ConclusionUmin} \\
E[g(Y_{max})]\xrightarrow[K\rightarrow\infty]{}{g}(\mu_K) \IEEEyessubnumber \label{ConclusionUmax} \\
E[g(Y_{min},Y_{max})]\xrightarrow[K\rightarrow\infty]{}{g}(\mu_1,\mu_K) \IEEEyessubnumber \label{ConclusionUmix}
\end{IEEEeqnarray}

However, in cases where the sample space is infinite and the conditions of \emph{Lemma} \ref{Lemma2} apply, the CLV and/or the CSV  converge to the \emph{mode}, i.e., the most probable value, of the extreme, as $K$ increases \cite{Gumbel}, and, as the variance of the GEV under these conditions is finite and independent of $K$, the following hold:
\begin{IEEEeqnarray}{cCl}
\frac{\mu_1}{E[Y_{min}]}\xrightarrow[K\rightarrow\infty]{}1 \IEEEyesnumber \IEEEyessubnumber \\ \label{eqlemma2minrep}
\frac{\mu_K}{E[Y_{max}]}\xrightarrow[K\rightarrow\infty]{}1 \IEEEyessubnumber \label{eqlemma2maxrep}
\end{IEEEeqnarray}
which, under the assumption $\forall \alpha: \left|\frac{\partial{^i}{g(\alpha)}}{\partial{\alpha}{^i}}\right|<\infty;i\in{\mathbb{N}}$, for sufficiently large $K$, yields
\begin{IEEEeqnarray}{cCl} \label{ConclusionULem2}
E[g(Y_{min})]\approx{g}(\mu_1) \IEEEyesnumber \IEEEyessubnumber \label{ConclusionUmin2} \\
E[g(Y_{max})]\approx{g}(\mu_K) \IEEEyessubnumber \label{ConclusionUmax2} \\
E[g(Y_{min},Y_{max})]\approx{g}(\mu_1,\mu_K) \IEEEyessubnumber \label{ConclusionUmix2}
\end{IEEEeqnarray}

Next, by implementing the set of Eqs. \eqref{ConclusionULem1} and \eqref{ConclusionULem2} directly on the different FIM expressions, the following approximations can be presented.
\begin{IEEEeqnarray}{cCl}
\widetilde{J}^{min}_{m,n}(\underline\theta)=N\cdot{g_{m,n}(\mu_1,\underline\theta)}\approx \IEEEyesnumber 
\sum_{i=1}^{N}E\left[g_{m,n}(z_i,\underline\theta)\right]=J^{min}_{m,n}(\underline\theta) \label{Japprx} \IEEEyessubnumber \label{Japprxminch2} \\
\widetilde{J}^{max}_{m,n}(\underline\theta)=N\cdot{g_{m,n}(\mu_K,\underline\theta)}\approx 
\sum_{i=1}^{N}E\left[g_{m,n}(z_i,\underline\theta)\right]=J^{max}_{m,n}(\underline\theta) \IEEEyessubnumber \label{Japprxmaxch2} \\
\widetilde{J}^{mix}_{m,n}(\underline\theta)=N\cdot{g_{m,n}(\mu_1,\mu_K,\underline\theta)}\approx\sum_{i=1}^{N}E\left[g_{m,n}(z_i,q_i,\underline\theta)\right]=J^{mix}_{m,n}(\underline\theta) \IEEEyessubnumber \label{Japprxmixch2}
\end{IEEEeqnarray}
which, given that $g(\cdot)$ is the known log-likelihood function, can be analytically expressed and solved.

Next, we define a set of matrices, $\mathcal{A}^i$, as follows.
\begin{IEEEeqnarray}{cCl}
J_{m,n}^{max}(\underline\theta)-J_{m,n}^{min}(\underline\theta)=\sum^{N}_{i=1}\mathcal{A}^i_{m,n} \IEEEyesnumber \label{Adef} 
\end{IEEEeqnarray}
where
\begin{IEEEeqnarray}{cCl}
\mathcal{A}^i_{m,n}= E\left\{\frac{\partial^2\log\left[f_{Y_{min_i}}(z_{i};{\theta})\right]}{\partial{\theta{_m}\theta{_n}}}\right.
\left.-\frac{\partial^2\log\left[f_{Y_{max_i}}(q_{i};{\theta})\right]}{\partial{\theta{_m}\theta{_n}}}\right\} \IEEEyesnumber
\end{IEEEeqnarray}
which, based on the approximations $\widetilde{J}^{min}(\underline\theta)$, $\widetilde{J}^{max}(\underline\theta)$, and $\widetilde{J}^{mix}(\underline\theta)$ (of Eq. \eqref{Japprx}), yields
{{
\begin{IEEEeqnarray}{cCl}
\mathcal{A}^i_{m,n}\underset{\underset{\forall{i}}{\uparrow}}{\approx}\mathcal{A}_{m,n}=\left.\frac{\partial^2\log\left[f_{Y_{min}}(z_i;{\theta})\right]}{\partial{\theta{_m}}\partial\theta{_n}}\right|_{z_i=\mu_1} - \IEEEyesnumber 
 \left.\frac{\partial^2\log\left[f_{Y_{max}}(q_i;{\theta})\right]}{\partial{\theta{_m}}\partial\theta{_n}}\right|_{q_i=\mu_K} \label{Agen}
\end{IEEEeqnarray}
which equals
\begin{IEEEeqnarray}{cCl}
\mathcal{A}_{m,n}=(K-1)\cdot\left[\frac{\partial{F_m}}{\partial\theta_n}\cdot\frac{\partial{F_m}}{\partial\theta_m}\cdot\frac{1}{(1-F_m)^2} \right. \label{AAA} 
\left.-\frac{\partial^2{F_m}}{\partial\theta{_n}\partial\theta{_m}}\cdot\frac{1}{1-F_m}+\frac{\partial{F_M}}{\partial\theta_n}\cdot\frac{\partial{F_M}}{\partial\theta_m}\cdot\frac{1}{(F_M)^2}-\frac{\partial^2F_M}{\partial\theta_{n}\partial\theta_m}\cdot\frac{1}{F_M}\right]+ \IEEEnonumber \\
+\frac{\partial^2f_m}{\partial\theta_{n}\partial\theta_{m}}\cdot\frac{1}{f_m}-\frac{\partial^2f_M}{\partial\theta_{n}\partial\theta_{m}}\cdot\frac{1}{f_M}
+{\frac{\partial{f_M}}{\partial\theta_n}\cdot\frac{\partial{f_M}}{\partial\theta_m}\cdot\frac{1}{(f_M)^2}}-{\frac{\partial{f_m}}{\partial\theta_n}\cdot\frac{\partial{f_m}}{\partial\theta_m}\cdot\frac{1}{(f_m)^2}} \IEEEyesnumber
\end{IEEEeqnarray}}}
where $F_m\equiv{F}_{X}(\mu_1;\theta)$, $F_M\equiv{F}_{X}(\mu_K;\theta)$, $f_m\equiv{f}_{X}(\mu_1;\theta)$, and $f_M\equiv{f}_{X}(\mu_K;\theta)$ (of Eq. \eqref{fgen}), which are known. Thus, the value of $\mathcal{A}$ can be directly calculated.

The value of $\mathcal{A}$ holds important information:
\begin{equation} \label{CasesA}
\begin{cases}
\mathcal{A} > 0, {\text{~~}} \Rightarrow {\text{~~}} \widetilde{J}^{max}(\underline\theta)>\widetilde{J}^{min}(\underline\theta) \\
\mathcal{A} < 0, {\text{~~}} \Rightarrow {\text{~~}} \widetilde{J}^{max}(\underline\theta)<\widetilde{J}^{min}(\underline\theta) \\
\mathcal{A} \approx 0, {\text{~~}} \Rightarrow {\text{~~}} \widetilde{J}^{min}(\underline\theta)\approx\widetilde{J}^{max}(\underline\theta) \\
\end{cases}
\end{equation}
Based on this analysis, the questioned inequalities of Eq. \eqref{MajorG} can be established for any relevant distribution of interest. Some specific cases are worth noting:
\begin{itemize}
\item{{Symmetric Distributions}: In the special case where the original observations, $\{x_i\}$ ; $\forall{i}$, follow a \emph{symmetric PDF}, in the sense that a value of $y$ exists, s.t. $f_X(y+\delta)=f_X(y-\delta)$ ; $\forall{\delta}\in{\mathcal{R}}$, either single-parameter (e.g., the \emph{uniform} distribution) or multi-parameter (e.g., the \emph{normal} distribution), the calculation of $\mathcal{A}$ yields $\mathcal{A}=0$. This is to be expected, since it can be proven\footnote{By substituting the variables s.t. $x_i\equiv-y_i$ ; $\forall{i}$, it can be shown that the general expressions of $J^{min}(\underline{\theta})$ and $J^{max}(\underline{\theta})$ are identical.} that for the symmetric case, $J^{min}(\underline{\theta})=J^{max}(\underline{\theta})$.}
\item{Non-Negative Distributions: Many real-world physical and socio-economical scenarios are modelled using PDFs of the form: $f_X(x<C)=0; C\geq{0}$. Examples of such cases are measurements of rain intensity,  earthquakes magnitudes, wind speeds, yearly income, and sun-solar flare intensity, among many other. In those cases, by implementing the methodology presented in \emph{Lemma 1} and \emph{Lemma 2}, it can be shown that for large value of $K$, (i.e., $K\rightarrow\infty$), $\mathcal{A}>0$. Thus, for many practical uses, given that $K$ is sufficiently large, one can consider the maxima to hold most of the information (with respect to parameter estimation). Furthermore, in cases where the distribution is non-positive, $f_X(x>C)=0; C\leq{0}$, such as when modeling the volume of water pending evaporation, the speed of free electrons in metals, etc., similar conclusions apply, where for sufficiently large $K$, $\mathcal{A}<0$.}
\end{itemize}


In order to demonstrate and analyze the proposed methodology, the original observations distribution will be assumed to be of the \emph{exponential} type. The exponential distribution was chosen due to various reasons: First, the exponential distribution is considered to be a good and solid model for various phenomena in many fields of interest (from rain-rate intensity \cite{salisu2010modeling} to income in the USA \cite{ExpEcon}). Second, 
the PDF of the maximum value taken from groups of $K$ samples of an iid distributed population that follows a distribution belonging to the exponential family, e.g., \emph{exponential}, \emph{Gamma}, $\chi^2$, or \emph{Normal}, will converge asymptotically to the \emph{Gumbel} distribution \cite{Gumbel}, as $K\rightarrow\infty$. Thus, the behavior of $\hat{\theta}_{max}$ and $\widetilde{J}^{max}(\theta)$ regarding the \emph{exponential} distribution may be used on other distributions of the \emph{exponential} family, once $K$ is sufficiently large. 

\section{Example: The Case of Exponential Distribution}
$z$ is said to follow the exponential distribution, with the parameter ${\theta}$, if $f_Z(z;\theta)=\frac{1}{\theta}e^{-\frac{1}{\theta}\cdot{z}}$, or, equivalently, $F_Z(z;\theta)=1-e^{-\frac{1}{\theta}\cdot{z}}$, so that $Pr\{Z\leq{z}\}=F_Z(z;\theta)$). The asymptotically optimal estimator $\hat{\theta}_{opt}$ can be implicitly expressed:
\begin{equation} \label{MLoptexp}
{\hat{\theta}}_{opt} = \frac{1}{N\cdot{K}}\sum\limits_{i=1}^{N\cdot{K}}{x_{i}}\IEEEyessubnumber
\end{equation}
where $x_i$ is the $i^{th}$ entry of the observation vector $\underline{x}$.
Similarly, the sub-optimal estimator $\hat{{\theta}}_L$ (of Eq. \eqref{MLEl}) can be expressed:
\begin{equation}
{\hat{\theta}}_{L} = \frac{1}{N\cdot{L}}\sum\limits_{p=1}^{N}\sum\limits_{l=1}^{L}{x_{p,l}} \label{lexp}
\end{equation}
where $x_{p,l}$ represents the $l^{th}$ sample within the $p^{th}$ group. Note that if $L=K$, then $\hat{\theta}_L\equiv\hat{\theta}_{opt}$, as expected.

Next, the extreme-based estimates for the exponential case can be directly formalized:
\begin{IEEEeqnarray}{cCl} \label{extremeexpML}
\hat{{\theta}}_{min}=\arg\max\limits_{{\theta}}\left[\left(\frac{K}{\theta}\right)^N{e^{-\frac{K}{\theta}\cdot\sum\limits_{i=1}^{N}{y_{min_i}}}}\right]= \IEEEnonumber 
\frac{\sum\limits_{i=1}^{N}y_{min_i}}{N} \IEEEyesnumber \label{minexpml} \IEEEyessubnumber \label{MLEminch3} \\
\hat{{\theta}}_{max}=\arg\max\limits_{{\theta}}\left[\left(\frac{K}{\theta}\right)^N{e^{-\sum\limits_{i=1}^{N}\frac{y_{max_i}}{\theta}}}\right. \left.\cdot\right. 
\left.\prod_{i=1}^{N}\left(1-e^{-\frac{{y_{max_i}}}{\theta}}\right)^{K-1}\right] \IEEEyessubnumber \label{MLEmaxch3} \\ \label{maxexpml}
\hat{{\theta}}_{mix}=\arg\max\limits_{{\theta}}\left[\left(\frac{K(K-1)}{\theta^2}\right)^N{e^{\frac{-\sum\limits_{i=1}^{N}\left(y_{min_i}+y_{max_i}\right)}{\theta}}}\right. \left.\right. 
\left.\cdot\prod_{i=1}^{N}\left(e^{-\frac{{y_{min_i}}}{\theta}}-e^{-\frac{{y_{max_i}}}{\theta}}\right)^{K-2}\right] \IEEEyessubnumber \label{MLEmixch3} \label{mixexpml}
\end{IEEEeqnarray}
where $y_{min_i}$ and $y_{max_i}$ are the minimum and maximum values of the $i^{th}$ ($1\leq{i}\leq{N}$) interval, respectively. It is noteworthy here that $\hat{\theta}_{min}$, which has a closed analytical expression, is equal to $\hat{\theta}_{L=1}$. This property of $\hat{\theta}_{min}$ is explained via the exponential minima properties, the distribution of which remains \emph{exponential}, with the parameter $\theta/K$ \cite{Gumbel}.

The corresponding FIMs of the estimators $\hat{\theta}_{opt}$, $\hat{\theta}_L$, and $\hat{\theta}_{min}$, $\hat{\theta}_{max}$, and $\hat{\theta}_{mix}$ (of Eqs. \eqref{MLoptexp},\eqref{lexp}, and \eqref{extremeexpML}), can now be formalized:
{\small{
\begin{IEEEeqnarray}{cCl}
{J^{opt}({\theta})}= \frac{N\cdot{K}}{\theta^2} \IEEEyesnumber \label{FIMexp} \IEEEyessubnumber \\ \label{FIMexpopt}
{J^{L}({\theta})}= \frac{N\cdot{L}}{\theta^2} \IEEEyessubnumber \label{FIMexpl} \\
{J^{min}({\theta})}= \frac{N}{\theta^2} \IEEEyessubnumber \label{FIMexpmin} \\
{J^{max}({\theta})}= -\frac{N}{\theta^2} + \frac{2}{\theta^3}\sum\limits_{i=1}^NE[y_{max_i}]+\frac{(K-1)}{\theta^3}  \IEEEyessubnumber \label{FIMexpmax} 
\cdot\sum\limits_{i=1}^NE\left[\frac{\left(\frac{y_{max_i}^2}{\theta}-2y_{max_i}\right)\cdot{}e^{\frac{-y_{max_i}}{\theta}}}{\left(1-e^{\frac{-y_{max_i}}{\theta}}\right)}\right.  \left.+\frac{\frac{y_{max_i}^2}{\theta}e^{\frac{-2y_{max_i}}{\theta}}}{\left(1-e^{\frac{-y_{max_i}}{\theta}}\right)^2}\right] \\
{J^{mix}({\theta})}= -\frac{2N}{\theta^2} + \frac{2}{\theta^3}\sum\limits_{i=1}^NE[y_{min_i}+y_{max_i}]+ \IEEEyessubnumber \label{FIMexpminmax} 
+\frac{(K-2)}{\theta^3}\sum\limits_{i=1}^NE\left[\frac{\left({y_{min_i}}e^{\frac{-y_{min_i}}{\theta}}-{y_{max_i}}e^{\frac{-y_{max_i}}{\theta}}\right)^2}{\theta\cdot\left(e^{\frac{-y_{min_i}}{\theta}}-e^{\frac{-y_{max_i}}{\theta}}\right)^2}\right. + \IEEEnonumber \\
\left. +\frac{\left({2y_{min_i}}-\frac{y_{min_i}^2}{\theta}\right)\cdot{e}^{\frac{-y_{min_i}}{\theta}}}{e^{\frac{-y_{min_i}}{\theta}}-e^{\frac{-y_{max_i}}{\theta}}}\right.
-\left.\frac{\left({2y_{max_i}}-\frac{y_{max_i}^2}{\theta}\right)\cdot{e}^{\frac{-y_{max_i}}{\theta}}}{e^{\frac{-y_{min_i}}{\theta}}-e^{\frac{-y_{max_i}}{\theta}}}\right]  
\end{IEEEeqnarray}}}
and, while $J^{opt}(\theta)$, $J^{L}(\theta)$, and $J^{min}(\theta)$ are explicitly presented, in order to simplify the expressions of $J^{max}(\theta)$ and $J^{mix}(\theta)$, we use the CLV (of Eq. \eqref{CLVch2}) and the CSV (of Eq. \eqref{CSVch2}) approximations. For the exponential distribution with the parameter $\theta$, the CSV, $\mu^{exp}_1$, and the CLV, $\mu^{exp}_K$, take the forms (for $K>1$)
\begin{IEEEeqnarray}{cCl} \label{CLV}
\mu^{exp}_{1}= \theta\cdot{ln\left({\frac{K}{K-1}}\right)} \IEEEyesnumber \IEEEyessubnumber \\
\mu^{exp}_{K}= \theta\cdot{ln({K})} \IEEEyessubnumber
\end{IEEEeqnarray}
and therefore, substituting $\mu^{exp}_1$ and $\mu^{exp}_K$ into Eqs. \eqref{FIMexpmax} and \eqref{FIMexpminmax} yields (for $K>2$)
\begin{IEEEeqnarray}{cCl}
\widetilde{J}^{max}(\theta)=\frac{N}{\theta^2}\cdot{}\left[\frac{K\cdot{ln}^2(K)}{K-1}-1\right] \IEEEyesnumber \label{FIMmaxapp} \IEEEyessubnumber \label{FIMapprxmaxch3} \\
\widetilde{J}^{mix}(\theta)=\frac{2N}{\theta^2}\cdot{}\left[\frac{(K-1)\cdot\left({ln}(K)-{ln}\left(\frac{K}{K-1}\right)\right)^2}{2(K-2)} 
+K\cdot{ln}\left(\frac{K}{K-1}\right)-1\right] \IEEEyessubnumber \label{FIMapprxmixch3}
\end{IEEEeqnarray}

It is worth noting that, although unnecessary, a comparison of the exact (analytical solvable) expression of $J^{min}(\theta)$, which is available for the exponential case, and the approximated expression, $\widetilde{J}^{min}(\theta)$ (via Eq. \eqref{Japprxminch2}), demonstrates the accuracy of the approximation. The approximation accuracy improves as $K$ increases, since $\widetilde{J}^{min}(\theta)\xrightarrow{K\rightarrow\infty}J^{min}(\theta)$). However, even for $K$ as small as $K=10$, the difference between $\widetilde{J}^{min}(\theta)$ and ${J}^{min}(\theta)$ is only $\approx{5\%}$. For $K=100$, the difference already drops to $\approx{0.5\%}$. 

Finally, by combining Eq. \eqref{MajorG} with Eqs. \eqref{FIMexp} and \eqref{FIMmaxapp}, the main relationship can be presented for the exponential case:
\begin{IEEEeqnarray}{cCl}
J^{L=1}(\theta){=}J^{min}(\theta)<\widetilde{J}^{max}(\theta) 
 <\widetilde{J}^{mix}(\theta)<J^{opt}(\theta){=}J^{{L=K}}(\theta) \IEEEyesnumber  \label{Main}
\end{IEEEeqnarray}

\subsection[Relative Contribution of the Min/Max Values for Estimation \\ Purposes]{The Relative Contribution of the Min/Max for Estimation \\ Purposes} \label{ContSimch2}
The relationship between the different FIMs (of Eq. \eqref{Main}) provides powerful insights. First, as can be seen in Eqs. \eqref{FIMexpl} and \eqref{FIMexpmin}, the minimum values-based estimator, $\hat{\theta}_{min}$, is identical to the estimator $\hat\theta_{L=1}$. This fact, which may be counter-intuitive, shows that fetching the minimum values themselves does not give an advantage over a random selection of a single value in the group. In other words, knowing the observed minimum value of each group for a given number of groups,
each of which constitutes $K$ iid exponential distributed observations, is identical (with respect to one parameter-estimation performance) to knowing a single randomly selected observation from each group. The situation is different regarding the maxima. Knowing the maximum observed values yields the maximum based estimator, $\hat\theta_{max}$, which achieves a better estimation performance than using only one value from the sequence $\left\{x_i\right\}^K_{i=1}$. Thus, when access to the observations is constrained such that only a single observation per group is allowed, which can be the minimum observed value, the maximum observed value, or a single randomly selected observation within the group, using the maximum observed values is the preferred choice from the parameter estimation accuracy point of view. Second, as expected, access to both the minimum \emph{and} the maximum values further enhances the estimation performance. The optimal estimator, $\hat\theta_{opt}$, obviously achieves the best performance.

Rewriting $J^{min}(\theta)$, $\widetilde{J}^{max}(\theta)$, and $\widetilde{J}^{mix}(\theta)$ in a more general form, with respect to $K$, yields
\begin{IEEEeqnarray}{cCl}
J^{min}(\theta)=\frac{N}{\theta^2}\cdot\mathcal{O}(1) \IEEEyesnumber \label{orderofmag} \IEEEyessubnumber \label{FIMtrueminch3} \\
\widetilde{J}^{max}(\theta)=\frac{N}{\theta^2}\cdot\mathcal{O}(ln^2(K)) \IEEEyessubnumber \label{FIMtruemaxch3} \\
\widetilde{J}^{mix}(\theta)=\frac{N}{\theta^2}\cdot\mathcal{O}(ln^2(K)) \IEEEyessubnumber \label{FIMtruemixch3}
\end{IEEEeqnarray}

where $\mathcal{O}(\cdot)$ stands for order of magnitude. While $J^{min}(\theta)$ is independent of $K$, both $\widetilde{J}^{max}(\theta)$ and $\widetilde{J}^{mix}(\theta)$ depend on $K$ in a similar manner. Furthermore, it can be shown that:
\begin{IEEEeqnarray}{cCl}
\widetilde{J}^{mix}(\theta)-\widetilde{J}^{max}(\theta) \leq J^{L=1}(\theta)=\frac{N}{\theta^2} \IEEEyesnumber \label{Converg} \IEEEyessubnumber \\
\text{where~~~} \widetilde{J}^{mix}(\theta)-\widetilde{J}^{max}(\theta) \underset{\underset{K\rightarrow\infty}{\uparrow}}{\rightarrow} \frac{N}{\theta^2} \IEEEyessubnumber
\end{IEEEeqnarray}
from which, it can be seen that the information stored in the maximum observed values increases as the group size $K$ increases. When the maximum observed values are available, the additional information from the minimum observed values is relatively small; indeed, as $K$ increases, the additional information from the minimum observed values increases, but is bounded by $N/\theta^2$, since, as $K\rightarrow\infty$, the added information converges to $N/\theta^2=J^{min}(\theta)$. This result is  expected, since it has been shown that for $K\rightarrow\infty$, the extremes, i.e., the minimum and the maximum values, become independent \cite{coles1999dependence}, and therefore,
\begin{equation} \label{conto1ch3}
J^{mix}(\theta)\underset{\underset{K\rightarrow\infty}{\uparrow}}{\rightarrow}J^{min}(\theta)+J^{max}(\theta)
\end{equation}

TABLE \ref{Jcomp} describes the behavior of the different (approximated) FIMs as a function of the group size $K$. The normalized values of the FIMs, expressed in Eqs. \eqref{FIMexp} and \eqref{FIMmaxapp}, correspond to the equivalent value of $L$ out of the $K$ samples in each group, defined as the \emph{L-equivalent} value (see Eqs. \eqref{CRLBla} and \eqref{CRLBlb}). As Eq. \eqref{Main} suggests $\forall{K}:J^{min}(\theta){=}J^{L=1}$, in Table \ref{Jcomp} the value of $J^{min}(\theta)$ is constant for all $K$. On the other hand, the estimates $\hat{\theta}_{max}$ and $\hat{\theta}_{mix}$ become more accurate as $K$ increases, with $\hat{\theta}_{mix}$ always being better. However, the difference between the performance of the estimates $\hat{\theta}_{max}$ and $\hat{\theta}_{mix}$ converges to $N/\theta^2$, which translates into adding only one sample, as can be seen in the column of $\Delta\equiv{\widetilde{J}}^{mix}(\theta)-\widetilde{J}^{max}(\theta)$. This is again, to be expected, as seen in Eq. \eqref{conto1ch3}.

\begin{table}[!ht]
\centering
\center{\caption{\emph{L-equivalent} values of $J^{opt}(\theta)$, $J^{min}(\theta)$, $\widetilde{J}^{max}(\theta)$, $\widetilde{j}^{mix}(\theta)$,  $J^{opt}(\theta)$, and $\Delta\equiv{\widetilde{J}}^{mix}(\theta)-\widetilde{J}^{max}(\theta)$, for selected values of $K$.} \label{Jcomp}}
\centering
\renewcommand{\arraystretch}{1.2}
\begin{tabular}{c c c c c c}
\hline \hline
$K$ & $\left[\frac{{J}^{\text{opt}}(\theta)}{N/\theta^2}\right]$ & $\left[\frac{J^{\text{min}}(\theta)}{N/\theta^2}\right]$ & $\left[\frac{\widetilde{J}^{\text{max}}(\theta)}{N/\theta^2}\right]$ & $\left[\frac{\widetilde{J}^{\text{mix}}(\theta)}{N/\theta^2}\right]$ & $\left[\frac{\Delta}{N/\theta^2}\right]$  \\
\hline \hline
5  & $5$ & $1$ & $2.238$ & $2.794$  & $0.556$\\
\hline
10 & $10$ & $1$ & $4.891$ & $5.539$  & $0.648$\\
\hline
25  & $25$ & $1$ & $9.793$ & $10.580$ & $0.787$\\
\hline
50  & $50$ & $1$ & $14.616$ & $15.482$ & $0.866$\\
\hline
100  & $100$ & $1$ & $20.422$ & $21.341$ & $0.919$\\
\hline
1000  & $1000$ & $1$ & $46.765$ & $47.752$ & $0.987$\\
\hline \hline
\end{tabular}
\end{table}

\subsection{Simulation Results} \label{simm}
In order to validate the relationship of Eq. \eqref{Main} and the insights presented in Table \ref{Jcomp}, the following simulation was designed.
\begin{itemize}
\item[(i)]{One hundred groups of $K$ iid exponentially distributed variables, with parameter $\theta=1$, were generated, i.e., $N=100$.}
\item[(ii)]{For each group, the minimum and the maximum values were logged. Thus, a vector containing 100 minimum observed values and a vector containing 100 maximum observed values were created.}
\item[(iii)]{Using these two vectors as inputs, the estimation processes of Eqs. \eqref{MLEminch3}, \eqref{MLEmaxch3}, and \eqref{MLEmixch3} were performed, in order to obtain the three estimates $\hat\theta_{min}$, $\hat\theta_{max}$, and $\hat\theta_{mix}$. In addition, for comparison, the optimal estimate, $\hat\theta_{opt}$, was calculated from the full dataset.}
\item[(iv)]{Steps (i)-(iii) were repeated 10000 times.}
\item[(v)]{Steps (i)-(iv) were repeated for  $K$ ranging from $K=5$ to $K=100$.}
\end{itemize}


In addition to the simulation, the analytical variances, i.e., the CRLBs, of $\hat\theta_{min}$ (exact), $\hat\theta_{max}$ (approximated), and $\hat\theta_{mix}$ (approximated) were calculated using Eqs. \eqref{FIMexp}, and \eqref{FIMmaxapp}.

The simulation results are depicted in {\textbf{Figure} \ref{b15to540}}. The calculated variances of $\hat\theta_{min}$, $\hat\theta_{max}$, and $\hat\theta_{mix}$, and their corresponding asymptotic variances, i.e., CRLBs, as a function of $K$, for $\theta=1$ and $N=100$ are shown. In addition, the calculated variance of $\hat\theta_{opt}$ and its corresponding asymptotic CRLB is drawn. As expected, the variance of ${\hat{\theta}}_{min}$ is independent of $K$, and equals $\theta^2/N=10^{-2}$), whereas the variances of ${\hat{\theta}}_{max}$ and ${\hat{\theta}}_{mix}$ are $K$-dependent, and decrease as $K$ increases. Note that for small values of $K$ ($\approx{K}<15$), the accuracy of the approximated CRLBs drops (due to Eq. \eqref{eqlemma2maxrep}). For larger values of $K$, the approximated CRLBs are rather accurate, and describe well the values of $var({\hat{\theta}}_{max})$ and $var({\hat{\theta}}_{mix})$. The simulated results are encouraging. First, note that the variance of $\hat\theta_{min}$ is constant, and follows the expression suggested in Eq. \eqref{FIMexpmin}. Second, as expected, it can be seen that, regardless of the value of $K$ (for $\approx{K}>15$), the performance of $\hat\theta_{max}$ and $\hat\theta_{mix}$ is described by the approximated expressions of Eqs. \eqref{FIMapprxmaxch3} and \eqref{FIMapprxmixch3}. Furthermore, it can be seen that both $\hat\theta_{max}$ and $\hat\theta_{mix}$ perform significantly better than $\hat\theta_{min}$, which agrees with Eqs. \eqref{FIMtrueminch3}, \eqref{FIMtruemaxch3}, and \eqref{FIMtruemixch3}. Thus, we conclude that the results of the simulation validate the suggested analysis. In addition, note that, even for relatively small values of $K\approx10$, the suggested approximations of $\widetilde{J}^{max}(\theta)$ and $\widetilde{J}^{mix}(\theta)$, although not as accurate as for higher values of $K$, are still usable, as their values are still relatively close to the simulated-based variances.

\begin{figure}[ht!]
\centering
\includegraphics[width=150mm]{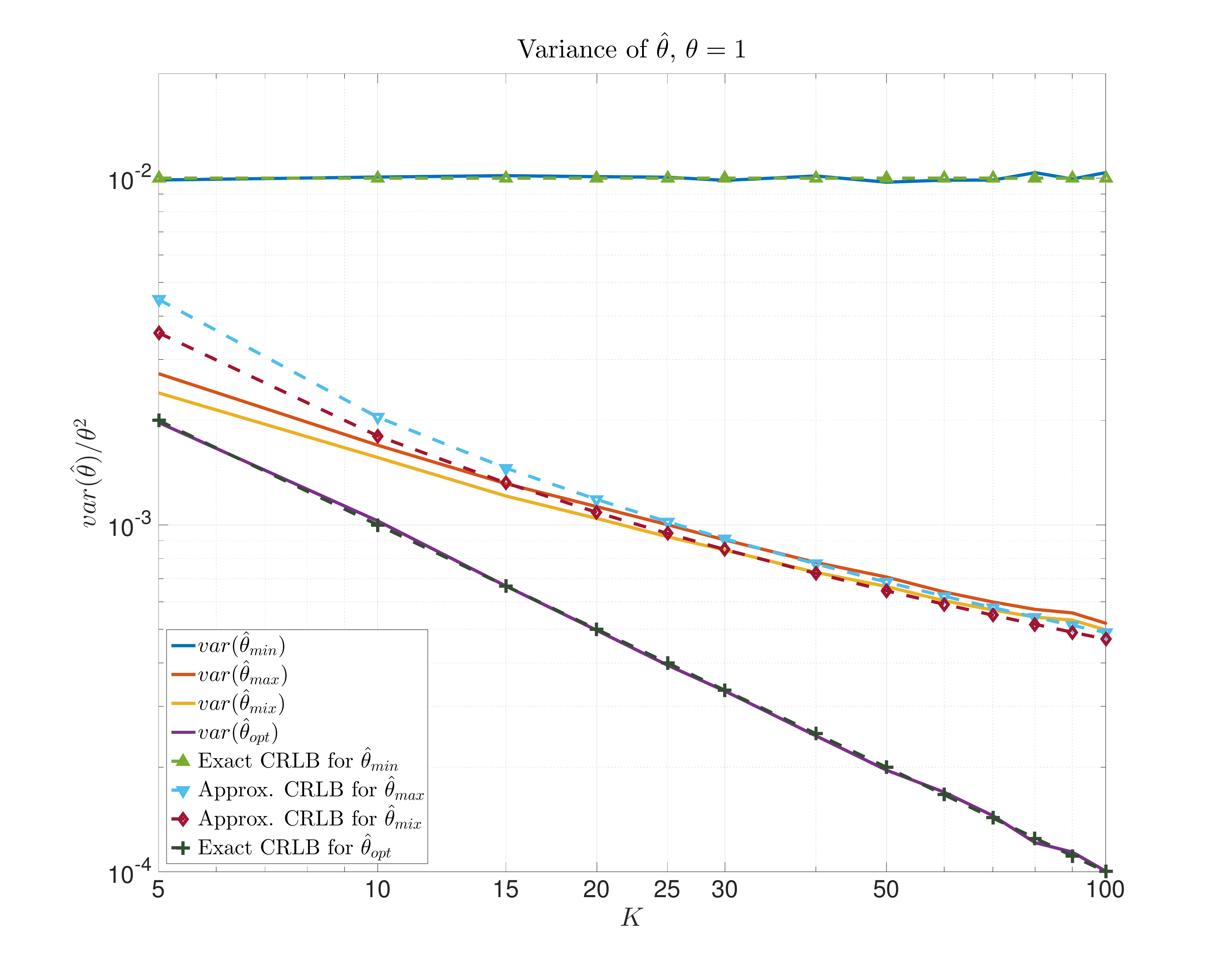}
\caption{Calculated variances of $\hat\theta_{min}$, $\hat\theta_{max}$, $\hat\theta_{mix}$, and their corresponding asymptotic variances, i.e., CRLBs, as a function of $K$, for $N=100$.} \label{b15to540}
\end{figure}

\subsection{Demonstration of the accuracy of the Fisher information matrix approximations}
The problem of estimating the original PDF parameter vector from a set of the upper or lower record values has been addressed in the field of record theory \cite{records,recordsparam1986}. Therefore, it is interesting to compare the proposed approximations with previously presented results.
The FIM properties, using only the minimum or only the maximum values of an exponentially distributed population, were discussed in \cite{hofmann2003}, and we now compare our findings with these properties.

Specifically, in \cite{hofmann2003} it was shown that the minimum-based estimate (for the exponential distribution) can be solved analytically, and has the \emph{L-equivalent} value of 1, regardless of $K$. This matches our findings, as we concluded that $J^{min}(\theta)=N/\theta^2$, and strengthens our CSV approximation, as we showed that $\widetilde{J}^{min}(\theta)\approx{J}^{min}(\theta)$.

Regarding the maximum based estimation, in \cite{hofmann2003} it was suggested that exponential order statistics properties be used to approximate $J^{max}(\theta)$ (see \cite{arnold1992first}, \emph{example 7.3.4}). The values of this approximation of $J^{max}(\theta)$, defined by ${J}^{max}_{2003}(\theta)$, for selected values of $K$, is presented, in an \emph{L-equivalent} normalized form, in Table \ref{Rcomp}, alongside the \emph{L-equivalent} normalized values of the proposed approximated CLV-based $\widetilde{J}^{max}(\theta)$, and the empirical variance of $\hat\theta_{max}$ calculated from the simulation described in this section, the results of which presented are in {\textbf{Figure} \ref{b15to540}.

As  can be seen in TABLE \ref{Rcomp}, although the values of ${J}^{max}_{2003}(\theta)$ are closer to the simulation results, the difference in the \emph{L-equivalent} values of $\widetilde{J}^{max}(\theta)$ and ${J}^{max}_{2003}(\theta)$ is less than $1.6$ for the entire range of $5\leq{K}\leq{200}$. This is encouraging, especially as the $\widetilde{J}^{max}(\theta)$ approximation is not confined to the exponential case and can be utilized for every distribution of interest.

Finally, it is worth noting that we could not compare the approximated results of $\widetilde{J}^{mix}(\theta)$, since the analysis of situations in which both of the extremes, i.e., the minimum \emph{and} the maximum values, are used at the same time has attracted less interest in the past, and no comparable publications that address the values of $\widetilde{J}^{mix}(\theta)$ are known to us. However, based on the presented comparison for both $\widetilde{J}^{min}(\theta)$ and $\widetilde{J}^{max}(\theta)$, in combination with the agreement of $\widetilde{J}^{mix}(\theta)$ with the simulation results (see Table \ref{Jcomp}),
the validity of the CSV- and the CLV-based approximations is promising.


\begin{table}[!ht]
\caption{Normalized values of $\widetilde{J}^{max}(\theta)$, ${J}^{max}_{2003}(\theta)$), and the simulation results, $var^{-1}(\hat\theta_{max})$, for selected values of $K$.} \label{Rcomp}
\centering
\renewcommand{\arraystretch}{1.2}
\begin{tabular}{c c c c}
\hline \hline
 & Approximated & Asymp. Exact & Empirical \\
$K$ & $\left[\frac{\widetilde{J}^{\text{max}}(\theta)}{N/\theta^2}\right]$ & $\left[\frac{{J}^{\text{max}}_{2003}(\theta)}{N/\theta^2}\right]$ & $\left[\frac{var^{-1}(\hat\theta_{max})}{N/\theta^2}\right]$ \\
\hline \hline
5  & $2.24$ & $3.66$ & $3.66$ \\
\hline
10 & $4.89$ & $5.86$ & $5.89$ \\
\hline
20 & $8.45$ & $8.87$ & $8.86$ \\
\hline
30 & $10.97$ & $11.05$ & $11.05$ \\
\hline
40 & $12.96$ & $12.79$ & $12.83$ \\
\hline
50 & $14.62$ & $14.26$ & $14.13$ \\
\hline
100 & $20.42$ & $19.45$ & $19.21$ \\
\hline
200 & $27.21$ & $25.63$ & $25.71$ \\
\hline \hline
\end{tabular}
\end{table}

It is worth noting that the approximations of $\widetilde{J}^{max}(\theta)$ and $\widetilde{J}^{mix}(\theta)$ are not the bounds of the actual values of the FIMs. For very large values of $K$, the approximations $\widetilde{J}^{max}(\theta)$ and $\widetilde{J}^{mix}(\theta)$ seem to be less accurate than for moderate values of $K$. This phenomenon is to be expected. In order to find an analytical solution to the presented problem, the expected values of the extremes were replaced with the CSV and the CLV, based on the approximations presented in Eqs. \eqref{ConclusionUmin}, \eqref{ConclusionUmax2}, and \eqref{ConclusionUmix2}. However, the CSV and the CLV converge to the \emph{modes} of the extremes (under the conditions of \emph{Lemma} \ref{Lemma2}), and not to the expected values. The difference between the mode and the expected value of the maximum observed value, is finite, and independent of $K$, whereas both the CLV and the expected value of $Y_{max}$, $E[Y_{max}]$, increase as $K$ increases. Thus, the difference between the CLV, $\mu^{exp}_K$, and the expected value, $E[Y_{max}]$, can be neglected for most values of $K$. However, as the \emph{Gumbel} distribution, to which $Y_{max}$ converges in distribution, is positively skewed, this difference introduces a small bias into the expressions of the FIM. This bias, although small, becomes non-negligible for very large values of $K$, as the variance of $\theta$ decreases. This fact is interesting, but unimportant for practical purposes, as it starts to affect the accuracy of the FIM approximations only for very large values of $K$ ($K>>100$), which is generally non-realistic for real-world usage, as for such large $K$, the information held in the extremes is only a fraction of the information held in the entire set of measurements (e.g., as presented in TABLE \ref{Jcomp}, for $K=1000$, estimating $\theta$ using the set of the maximum values is equivalent (performance wise) to an estimate based on 47 measurements per interval, which is a merely $4.7\%$ of the available original measurements).

\section{Conclusion}
This paper deals with the information in extreme values with respect to parameter estimation. We presented a new approach which uses the Characteristic Values of the extremes, by which we were able to establish simple and analytical solvable approximations to the expressions of the Fisher Information Matrices of the estimates based on either the minimum, the maximum, or the minimum \emph{and} the maximum measurements.  

Based on these approximations, we designed a new tool to evaluate the accuracy of estimates which use only the minimum or the maximum measurements, relative to the optimal estimation, which uses the entire dataset of measurements. We showed that the presented methodology gives simple and practical expressions, which are solvable, yet still capable to approximate the performance of the estimates. Furthermore, we expanded the proposed practical expressions and were able to analyse the performance of the estimate based on both extremes (i.e., the minimum \emph{and} the maximum values, combined).

We demonstrated our tool and performed an analysis for the case of the exponential distribution, which is an important example, as it is used to model many naturally occurring phenomena, especially for environmental monitoring. This fact makes the exponential case important for real-world applications. In addition, the exponential case is one of the few examples that were studied previously, and allowed us to perform a detailed comparison of our approximations with past results. Moreover, as the exponential distribution share commonality with many other distributions (in which, the extreme behaviour is similar), such as the log-normal and the gamma, our performed analysis on the exponential case, which demonstrates the importance of the maximum values rather than the minimum values (with respect to the accuracy of the outcome estimates) can be directly applied to other distribution families.

To conclude, we believe that the new approach and tools we have presented in this paper hold the potential to be used in many current and future applications, especially in the emerging field of IoT sensing, where the logging and transmitting of less observations is desired, for instance, where the cost per additional measurement is high.

\singlespacing

\bibliography{mybibJan2017Drive}

\begin{thebibliography}{10}

\bibitem{Key}
S.~M. Kay, ``Fundamentals of statistical signal processing,'' {\em PTR
  Prentice-Hall, Englewood Cliffs, NJ}, 1993.

\bibitem{wang2011}
J.~Wang, C.~Chan, and Y.~Wu, ``The distribution of annual maximum earthquake
  magnitude around taiwan and its application in the estimation of catastrophic
  earthquake recurrence probability,'' {\em Natural hazards}, vol.~59, no.~1,
  pp.~553--570, 2011.

\bibitem{walshaw1994getting}
D.~Walshaw, ``Getting the most from your extreme wind data: a step by step
  guide,'' {\em Journal of Research of the National Institute of Standards and
  Technology}, vol.~99, pp.~399--399, 1994.

\bibitem{MaxTemp1}
G.~Baskerville and P.~Emin, ``Rapid estimation of heat accumulation from
  maximum and minimum temperatures,'' {\em Ecology}, pp.~514--517, 1969.

\bibitem{ExtremeRainEx1}
D.~Koutsoyiannis, ``Statistics of extremes and estimation of extreme rainfall:
  Empirical investigation of long rainfall records,'' {\em Hydrological
  Sciences Journal}, vol.~49, no.~4, 2004.

\bibitem{EnvSN}
P.~Corke, T.~Wark, R.~Jurdak, W.~Hu, P.~Valencia, and D.~Moore, ``Environmental
  wireless sensor networks,'' {\em Proceedings of the IEEE}, vol.~98, no.~11,
  pp.~1903--1917, 2010.

\bibitem{RemkoCountry2}
A.~Overeem, H.~Leijnse, and R.~Uijlenhoet, ``Country-wide rainfall maps from
  cellular communication networks,'' {\em Proceedings of the National Academy
  of Sciences}, vol.~110.8, pp.~2741--2745, 2013.

\bibitem{Eric10S}
E.~website, ``http://www.ericsson.com,''

\bibitem{yonidiss}
J.~Ostrometzky, ``Statistical signal processing of extreme attenuation
  measurements taken by commercial microwave links for rain monitoring,'' {\em
  Ph.D Dissertation, Tel Aviv University}, 2017.

\bibitem{lemire2006streaming}
D.~Lemire, ``Streaming maximum-minimum filter using no more than three
  comparisons per element,'' {\em arXiv preprint cs/0610046}, 2006.

\bibitem{sedraandsmith}
A.~S. Sedra and K.~C. Smith, {\em Microelectronic circuits}, vol.~1.
\newblock Oxford university press, 1998.

\bibitem{ClimateExtreme92}
G.~MacDonald, ``Statistics of extreme events with application to climate,''
  {\em JSR}, vol.~90, p.~30S, 1992.

\bibitem{Flood2007}
D.~Norbiato, M.~Borga, M.~Sangati, and F.~Zanon, ``Regional frequency analysis
  of extreme precipitation in the eastern italian alps and the august 29, 2003
  flash flood,'' {\em Journal of Hydrology}, vol.~345, no.~3, pp.~149--166,
  2007.

\bibitem{Katz3}
R.~W. Katz, M.~B. Parlange, and P.~Naveau, ``Statistics of extremes in
  hydrology,'' {\em Advances in water resources}, vol.~25, no.~8,
  pp.~1287--1304, 2002.

\bibitem{coles2003anticipating}
S.~Coles and L.~Pericchi, ``Anticipating catastrophes through extreme value
  modelling,'' {\em Journal of the Royal Statistical Society: Series C (Applied
  Statistics)}, vol.~52, no.~4, pp.~405--416, 2003.

\bibitem{Gumbel}
E.~J. Gumbel, {\em Statistics of Extremes}.
\newblock Dover; ISBN 0-486-43604-7, 1958.

\bibitem{FisherT}
R.~Fisher and L.~Tippett, ``Limiting forms of the frequency distribution of the
  largest or smallest member of a sample,'' in {\em Mathematical Proceedings of
  the Cambridge Philosophical Society}, vol.~24, pp.~180--190, Cambridge Uni.
  Press, 1928.

\bibitem{walshaw2013}
D.~Walshaw, ``Generalized extreme value distribution,'' {\em Encyclopedia of
  Environmetrics}, 2013.

\bibitem{GEVEstimator}
E.~Martins and R.~Stedinger, ``Generalized maximum-likelihood generalized
  extreme-value quantile estimators for hydrologic data,'' {\em Water Resources
  Research}, vol.~36, no.~3, pp.~737--744, 2000.

\bibitem{hosking1985}
J.~Hosking, J.~R. Wallis, and E.~F. Wood, ``Estimation of the generalized
  extreme-value distribution by the method of probability-weighted moments,''
  {\em Technometrics}, vol.~27, no.~3, pp.~251--261, 1985.

\bibitem{Dekkers}
A.~L.~M. Dekkers, J.~H.~J. Einmahl, and L.~D. Haan, ``A moment estimator for
  the index of an extreme value distribution,'' {\em The Annals of Statistics},
  vol.~17/4, pp.~1833--1855, 1989.

\bibitem{records}
B.~C. Arnold, N.~Balakrishnan, and H.~N. Nagaraja, {\em Records}, vol.~768.
\newblock John Wiley \& Sons, 2011.

\bibitem{recordsparam1986}
F.~J. Samaniego and L.~R. Whitaker, ``On estimating population characteristics
  from record-breaking observations. i. parametric results,'' {\em Naval
  Research Logistics Quarterly}, vol.~33, no.~3, pp.~531--543, 1986.

\bibitem{ahmadi2001fisher}
J.~Ahmadi and N.~R. Arghami, ``On the fisher information in record values,''
  {\em Metrika}, vol.~53, no.~3, pp.~195--206, 2001.

\bibitem{hofmann2003}
G.~Hofmann and H.~Nagaraja, ``Fisher information in record data,'' {\em
  Metrika}, vol.~57, no.~2, pp.~177--193, 2003.

\bibitem{recordsreview2009}
G.~Zheng, N.~Balakrishnan, and S.~Park, ``Fisher information in ordered data: a
  review,'' {\em Stat Interface}, vol.~2, pp.~101--113, 2009.

\bibitem{SAM2014}
J.~Ostrometzky and H.~Messer, ``Accumulated rainfall estimation using maximum
  attenuation of microwave radio signal,'' {\em IEEE, SAM}, pp.~193--196, 2014.

\bibitem{EM}
A.~P. Dempster, N.~M. Laird, and D.~B. Rubin, ``Maximum likelihood from
  incomplete data via the em algorithm,'' {\em Journal of the royal statistical
  society. Series B (methodological)}, pp.~1--38, 1977.

\bibitem{LibermanTh}
Y.~Liberman, ``Optimal recovery of rain field maps using wireless sensors
  network,'' {\em M.Sc Thesis}, vol.~Tel-Aviv University, 2013.

\bibitem{RemkoCountry}
A.~Overeem, H.~Leijnse, and R.~Uijlenhoet, ``Measuring urban rainfall using
  microwave links from commercial cellular communication networks,'' {\em Water
  Resources Research}, vol.~47, p.~W12505, 2011.

\bibitem{Yoni1}
J.~Ostrometzky, D.~Cherkassky, and H.~Messer, ``Accumulated precipitation
  estimation using measurements from multiple microwave links,'' {\em Adv.
  Meteorology}, vol.~Special Issue (PRES), 2015.

\bibitem{coles1999dependence}
S.~Coles, J.~Heffernan, and J.~Tawn, ``Dependence measures for extreme value
  analyses,'' {\em Extremes}, vol.~2, no.~4, pp.~339--365, 1999.

\bibitem{papoulis2002probability}
A.~Papoulis and S.~U. Pillai, {\em Probability, random variables, and
  stochastic processes}.
\newblock Tata McGraw-Hill Education, 2002.

\bibitem{salisu2010modeling}
D.~Salisu, S.~Supiah, A.~Azmi, {\em et~al.}, ``Modeling the distribution of
  rainfall intensity using hourly data,'' {\em American Journal of
  Environmental Sciences}, vol.~6, no.~3, pp.~238--243, 2010.

\bibitem{ExpEcon}
A.~Dr{\u{a}}gulescu and V.~M. Yakovenko, ``Evidence for the exponential
  distribution of income in the usa,'' {\em The European Physical Journal
  B-Condensed Matter and Complex Systems}, vol.~20, no.~4, pp.~585--589, 2001.

\bibitem{arnold1992first}
B.~C. Arnold, N.~Balakrishnan, and H.~N. Nagaraja, {\em A first course in order
  statistics}, vol.~54.
\newblock Siam, 1992.

\end{thebibliography}
\bibliographystyle{ieeetr}

\end{document}